\documentclass[11pt]{article}
\usepackage{geometry}
\usepackage{amsmath, amssymb, amsthm}
\usepackage{graphicx}
\usepackage{bm}
\usepackage[linesnumbered,ruled]{algorithm2e}
\SetKwRepeat{Do}{do}{while}
\usepackage{fullpage}
\usepackage{enumerate}
\usepackage{algorithmic}
\newtheorem{theorem}{Theorem}
\newtheorem{lemma}{Lemma}

\newtheorem{definition}{Definition}

\newtheorem{example}{Example}
\DeclareMathOperator{\dist}{dist}
\DeclareMathOperator{\elm}{elm}
\DeclareMathOperator{\sgn}{sgn}
\DeclareMathOperator{\SF}{\sigma}

\title{A Necessary Condition for Connectedness of Solutions to Integer Linear Systems\thanks{%
This work was supported by JSPS KAKENHI Grant Number JP24K14825.}
}

\author{Takasugu Shigenobu\thanks{%
Graduate School of Mathematics, 
Kyushu University.
{\ttfamily shigenobu.takasugu.563@s.kyushu-u.ac.jp}}
\and 
Naoyuki Kamiyama\thanks{%
Institute of Mathematics for Industry, 
Kyushu University.
{\ttfamily kamiyama@imi.kyushu-u.ac.jp}}
}
\date{}

\begin{document}

\maketitle 

\begin{abstract}
An integer linear system is a set of inequalities with integer constraints.
The solution graph of an integer linear system is 
an undirected graph defined on the set of feasible solutions to the integer linear system.
In this graph, a pair of feasible solutions is
connected by an edge if the Hamming distance between them is one.
In this paper, we consider a condition
under which the solution graph is connected for any right-hand side vector.
First, we prove that if the solution graph is connected for any right-hand side vector,
then the coefficient matrix of the system does not contain
some forbidden pattern as a submatrix.
Next, we prove that if at least one of
(i) the number of rows is at most 3, 
(ii) the number of columns is at most 2, 
(iii) the number of rows is 4 and
the number of columns is 3 holds,
then the condition that the coefficient matrix of the system does not contain the forbidden pattern 
is a sufficient condition under which the solution graph is connected for any right-hand side vector.
This result is stronger than a known necessary and sufficient condition 
since the set of coefficient matrix dimensions is strictly larger.

\end{abstract}  

\section{Introduction}

An integer linear system is a set of $m$ inequalities in $n$ variables with integer constraints.
More concretely, in an integer linear system, we are given an 
$m \times n$ real coefficient matrix $A$,
an $m$-dimensional real vector $b$, and a positive integer $d$.
A feasible solution of the integer linear system is an $n$-dimensional integer vector $x \in \{0,1,\dots,d\}^n$ such that $A x \ge b$.
It is known that the Boolean satisfiability problem is 
a special case of the problem of finding a solution to a given integer linear system.
The solution graph of an integer linear system is 
an undirected graph defined on the set of feasible solutions to the integer linear system.
In the solution graph, a pair of feasible solutions is connected by an edge 
if the Hamming distance between them is one.
That is, in the solution graph, 
a pair of feasible solutions that differ at exactly one position is connected.

In this paper, we consider a condition 
under which the solution graph of an integer linear system is connected for any right-hand side vector.
Kimura and Suzuki~\cite{KIMURA202188} proved 
that if the coefficient matrix has an elimination ordering,
then the solution graph is connected  for any right-hand side vector.
(See Section \ref{section:Pre} for the definition of an elimination ordering.)
However, Shigenobu and Kamiyama~\cite{shigenobuCOCOA2023} proved that 
this condition is not a necessary condition in general.
Specifically, Shigenobu and Kamiyama~\cite{shigenobuCOCOA2023} considered 
this problem from the viewpoint of the numbers of rows and columns of the coefficient matrix.
First, they proved that
if $m \ge 4$ and $n \ge 3$,
then there exists a matrix such that
the solution graph of the corresponding linear system is 
connected for any right-hand side vector,
but the matrix does not have an elimination ordering.
In addition, they proved that,
in the case where at least one of 
(i) $m \le 3$, 
(ii) $n \le 2$ holds,
if the solution graph is connected for any right-hand side vector,
then the coefficient matrix has an elimination ordering.
That is, in this case, the condition that the coefficient matrix has an elimination ordering is 
a necessary and sufficient condition 
under which the solution graph is connected for any right-hand side vector.

The goal of this paper is to obtain a necessary and sufficient condition that is stronger 
than that proposed by \cite{shigenobuCOCOA2023}.
First, we prove that if the solution graph is connected for any right-hand side vector,
then the coefficient matrix does not contain
some forbidden pattern as a submatrix,
i.e., we propose a novel necessary condition under which the solution graph is connected for any right-hand side vector.
Next, we prove that if at least one of 
(i) $m \le 3$,
(ii) $n \le 2$, 
(iii) $m = 4$ and $n = 3$
holds,
then the condition that the coefficient matrix does not contain the forbidden pattern 
is a sufficient condition 
under which the solution graph is connected for any right-hand side vector.
This result is stronger than the necessary and 
sufficient condition proposed by \cite{shigenobuCOCOA2023}.
The highlight of the proof of the second result is the proof of the case (iii).
In this case, we prove that 
if the coefficient matrix does not contain the forbidden pattern,
then the sign pattern of the matrix is uniquely determined.
Furthermore, we prove that
the matrix having this sign pattern is connected for any right-hand side vector
in the same way as a similar result 
in \cite{shigenobuCOCOA2023} for a $\{0,\pm 1 \}$-matrix.
Thus, we can obtain the proof of the case (iii) by combining these facts.
It is open whether the necessary condition proposed in this paper 
is a sufficient condition in general.

A solution graph is closely related to a reconfiguration problem~\cite{IDHPSUU11,N18}.
A reconfiguration problem is the problem of deciding which
a given initial solution can be transformed 
into a given target solution by step-by-step operation
(see, e.g.,~\cite{IDHPSUU11,N18}).
In the reconfiguration problem for an integer linear system, 
one of the natural reconfiguration rules
is to change only one element of the solution at each step.
The answer for  
the reconfiguration problem for an integer linear system is yes 
if and only if the target solution is reachable 
from the initial solution in the solution graph.
Thus, if the solution graph is connected,
then the answer for the reconfiguration problem for any pair 
of an initial solution and a target solution is yes.
Kimura and Suzuki~\cite{KIMURA202188} proved that 
the computational complexity of the reconfiguration problem 
for an integer linear system has a trichotomy.
Furthermore, it is known that 
the computational complexity of the reconfiguration problem 
for the Boolean satisfiability problem 
has a dichotomy~\cite{doi:10.1137/07070440X,Schwerdtfeger14}.

\section{Preliminaries}\label{section:Pre}
In this paper, let $\mathbb{R}$ denote the set of real numbers.
For all positive integers $n$, we define $[n] := \{ 1, 2, \ldots , n \}$.

\begin{definition}[Sign function]
    The sign function $\sgn : \mathbb{R} \to \{-1,0,1\}$ is 
    defined 
    as follows. 
    If $x < 0$, then we define $\sgn(x) := -1$. 
    If $x = 0$, then we define $\sgn(x) := 0$. 
    If $x > 0$, then we define $\sgn(x) := 1$. 
\end{definition}

First, we define an integer linear system and its solution graph.
Throughout this paper, we fix a positive integer $d$, and 
we define $D := \{0,1,\dots,d\}$. 
The set $D$ represents the domain of a variable in an integer 
linear system. 

\begin{definition}[Integer linear system]
    An integer linear system has 
    a coefficient matrix $A \in \mathbb{R}^{[m] \times [n]} $ and 
    a vector $b \in \mathbb{R}^{[m]}$.
    This integer linear system is denoted by $(A,b)$.
    A feasible solution to $(A,b)$ is a vector $x \in D^{[n]}$ such that $A x \ge b$. 
    The set of feasible solutions to $(A,b)$ is denoted by $R(A,b)$.
\end{definition}

\begin{definition}[Hamming distance]
    Define the function $\dist \colon \mathbb{R}^{[n]} \times \mathbb{R}^{[n]} \to \mathbb{R}$ by 
    \(
    \dist(x,y) :=  \left| \left\{ j \in \{1,\ldots,n\} : x_j \neq y_j \right\} \right|
    \)
    for all vectors $x,y \in \mathbb{R}^{[n]}$. 
    This function is called 
    the Hamming distance on $\mathbb{R}^{[n]}$.
\end{definition}

\begin{definition}[Solution graph]
    Let $R$ be a subset of $D^{[n]}$.
    We define $V(R) := R$ and $E(R) := \{ \{ x, y \} : x,y \in V(R) , \dist(x, y) = 1 \}$.
    We define the solution graph $G(R)$ 
    as the undirected graph with the vertex set $V(R)$ and the edge set $E(R)$.
    For each integer linear system $(A,b)$, we define $G(A,b) := G(R(A,b))$. 
\end{definition}

In this paper, we consider a condition 
under which the solution graph of an integer linear system 
is connected for any right-hand side vector.
Since we often use this condition,
we name this condition universal connectedness.

\begin{definition}[Universal connectedness]
    Let $A$ be a matrix in $\mathbb{R}^{[m]\times [n]}$.
    The matrix $A$ is said to be universally connected
    if, for all vectors $b \in \mathbb{R}^{[m]}$,
    $G(A,b)$ is connected.
\end{definition}

Next, we define a forbidden pattern and an elimination ordering.
These concepts are related to the sign pattern of a matrix.

\begin{definition}[Forbidden pattern]\label{def:FP} 
    Let $A = (a_{ij})$ be a matrix in $\mathbb{R}^{[m]\times [n]}$.
    Let $(I,J)$ be an element in $2^{[m]} \times 2^{[n]}$ satisfying $|I| = |J| \neq 0$.
    Define $\lambda := |I|$.
    Then $(I,J)$ is called a forbidden pattern in $A$ 
    if there exist orderings
    $I = \{i_1, i_2, \ldots, i_{\lambda}\}$ and $J = \{j_1, j_2, \ldots, j_{\lambda}\}$
    such that 
    \begin{itemize}
    \item
    $a_{ij_{\ell}} = 0$
    for all integers $\ell \in [\lambda]$ and all integers 
    $i \in I \setminus \{i_{\ell},i_{\ell+1}\}$, and 
    \item
    $a_{i_{\ell} j_{\ell}} a_{i_{\ell+1} j_{\ell}} < 0$
    for all integers $\ell \in [\lambda]$, 
    \end{itemize}
    where we define $i_{\lambda+1} := i_1$. 
\end{definition}

By Definition~\ref{def:FP}, if $(I,J)$ is a forbidden pattern, then $|I| = |J| \ge 2$.

\begin{example}
    Assume that a matrix $A \in \mathbb{R}^{[m] \times [n]}$ has a forbidden pattern
    $(I,J) = (\{ i_1, i_2, i_3, i_4 \},$ $ \{ j_1, j_2, j_3, j_4 \}) \in 2^{[m]} \times 2^{[n]}$.
    The submatrix $A^{F} = (a^{F}_{k \ell}) \in \mathbb{R}^{I \times J}$ 
    of $A$ defined by $a^{F}_{k \ell} := a_{i_{k} j_{\ell}}$
    has the following sign pattern 
    \[
    A^{F} = 
    \begin{pmatrix}
        + & 0 & 0 & - \\
        - & + & 0 & 0 \\
        0 & - & + & 0 \\
        0 & 0 & - & +
    \end{pmatrix},
    \]
    or a sign pattern obtained by exchanging the signs of non-zero elements in some columns,
    where $+$ means a positive number and $-$ means a negative number.
\end{example}

Next, we define an elimination ordering.
An elimination ordering is an order of column vectors
defined by this sign of each element of a matrix.

\begin{definition}[Elimination]\label{def:elimination}
    Let $A = (a_{ij})$ be a matrix in $\mathbb{R}^{[m]\times [n]}$.
    Let $j$ be an integer in $[n]$.
    We say that $A$ can be eliminated at the column $j$ if it satisfies at least one of the following conditions.
    \begin{enumerate}[{\rm (i)}]
        \item
        For all integers $i \in [m]$, if $ a_{ij} > 0$, then $a_{ij^{\prime}}=0$ for all integers $j^{\prime} \in [n] \setminus \{j\}$.
        \item
        For all integers $i \in [m]$, if $ a_{ij} < 0$, then $a_{ij^{\prime}}=0$ for all integers $j^{\prime} \in [n]\setminus \{j\}$.
    \end{enumerate}
\end{definition}

\begin{definition}[Eliminated matrix]
    Let $A$ be a matrix in $\mathbb{R}^{[m]\times [n]}$.
    Let $J$ be a subset of $[n]$.
    Then we define the eliminated matrix 
    $\elm(A,J) \in \mathbb{R}^{ [m] \times ([n] \setminus J) }$
    as the submatrix of $A$ whose index set of columns is $[n] \setminus J$.
\end{definition}

\begin{definition}[Elimination ordering]
    Let $A$ be a matrix in $\mathbb{R}^{[m]\times [n]}$.
    Let $S = (j_1,j_2,\dots, j_n)$ be a sequence of 
    integers in $[n]$.
    Then $S$ is called an elimination ordering of $A$ if, for all integers $t \in [n]$, 
    $\elm(A,\{j_1,j_2,\dots,j_{t-1}\})$ can be eliminated at the column $j_t$.
\end{definition}

In this paper, we need the following lemma.
In this lemma, we have to define the following algorithm.

\begin{lemma}[{\cite[Lemma 2]{shigenobuCOCOA2023}}]\label{lem:expansion_lemma}
    Let $A$ be a matrix in $\mathbb{R}^{[m]\times [n]}$.
    Define $A^r$ as a submatrix of $A$ 
    whose index set of columns is $\Delta$ defined by Algorithm~\ref{alg}.
    If $A^r$ is not universally connected,
    then $A$ is not universally connected.
\end{lemma}

\begin{algorithm}[ht]
\caption{Algorithm for defining $\Delta$.}
\label{alg}
\begin{algorithmic}[1]    
\STATE $E \leftarrow \emptyset$
\WHILE{$\elm(A,E)$ can be eliminated at some column}
\STATE Find an index $j \in [n] \setminus E$ at which $\elm(A,E)$ can be eliminated.
\STATE $E \leftarrow E \cup \{ j \}$ 
\ENDWHILE
\STATE Output $[n] \setminus E$ as $\Delta$.
\end{algorithmic}
\end{algorithm}

The contribution of this paper is summarized as follows.

\begin{theorem} \label{theorem:necessary}
    Let $A$ be a matrix in $\mathbb{R}^{[m] \times [n]}$.
    If $A$ is universally connected,
    then $A$ does not have a forbidden pattern.
\end{theorem}

\begin{theorem} \label{theorem:sufficient}
    Let $A$ be a matrix in $\mathbb{R}^{[m]\times [n]}$.
    Suppose that we have at least one of
    {\rm (i)} $m \le 3$, 
    {\rm (ii)} $n \le 2$, 
    {\rm (iii)} $m = 4$ and $n = 3$.
    If $A$ does not have a forbidden pattern,
    then $A$ is universally connected.
\end{theorem}

\section{Proof of Theorem~\ref{theorem:necessary}}

In this section, we prove Theorem~\ref{theorem:necessary}.
Define the function 
$\SF: \mathbb{R} \to \mathbb{R}$ by
\[
    \SF(a) := 
    \frac{1 + \sgn(a)}{2} 
\]
for all real numbers $a \in \mathbb{R}$.
If $a > 0$ (resp. $a < 0$), then $\SF(a) = 1$ (resp. $\SF(a) = 0$).
In this proof, we need the following lemmas.

\begin{lemma}[{\cite[Theorem 2]{shigenobuCOCOA2023}}]\label{lem:m<3_and_n<2_not_EO_->_not_UC}
    Let $A$ be a matrix in $\mathbb{R}^{[m] \times [n]}$.
    Suppose that we have at least one of
    {\rm (i)} $m \le 3$, 
    {\rm (ii)} $n \le 2$.
    If $A$ does not have an elimination ordering,
    then $A$ is not universally connected.
\end{lemma}

\begin{lemma}\label{lem:FP->notEO}
    Let $A$ be a matrix in $\mathbb{R}^{[m]\times [n]}$.
    If $A$ has a forbidden pattern, 
    then $A$ does not have an elimination ordering.
\end{lemma}
\begin{proof}
Assume that $A = (a_{ij})$ has a forbidden pattern.
Let $(I,J)$ be a forbidden pattern that minimizes $|I| \,(= |J|)$.
We define $\lambda := |I|$.

First, we consider the case where $\lambda = 2$.
We assume that $I = \{i_1,i_2\}$ and 
$J = \{j_1,j_2\}$.
Definition \ref{def:FP} implies that
$a_{i_1 j_1} a_{i_2 j_1} < 0$ and $a_{i_1 j_2} a_{i_2 j_2} < 0$.
Thus, $A$ cannot be eliminated at the columns $j_1$ or $j_2$.

Next, we consider the case where $\lambda \ge 3$.
In order to prove this case by contradiction,
we assume that $A$ has an elimination ordering.
Fix an elimination ordering of $A$.
Let $j$ be the first column in $J$ 
in the elimination ordering.
By Definition \ref{def:FP},
there exist distinct integers $i^+,i^- \in I$ 
such that $a_{i^+ j} > 0$ and $a_{i^- j} < 0$.
By the definition of $j$,
we have at least one of
(i) $a_{i^+ j^{\prime}} = 0$ for all integers $j^{\prime} \in J \setminus \{j\}$, 
(ii) $a_{i^- j^{\prime}} = 0$ for all integers $j^{\prime} \in J \setminus \{j\}$.

Assume that the case (i) (resp. the case (ii)) holds.
We define $J^{\prime} := J \setminus \{j\}$ 
and $I^{\prime} := I \setminus \{i^+\}$
(resp.\, $I^{\prime} := I \setminus \{i^-\}$).
For all integers $j^{\prime} \in J^{\prime}$, 
there exist distinct integers $i^{\prime}_1, i^{\prime}_2 \in I^{\prime}$ 
such that $a_{i^{\prime}_1 j^{\prime}} a_{i^{\prime}_2 j^{\prime}} < 0$
since $(I,J)$ is a forbidden pattern
and $a_{i^+ j^{\prime}} = 0$ (resp.\, $a_{i^- j^{\prime}} = 0$)
for all integers $j^{\prime} \in J^{\prime}$.
Thus, $(I^{\prime}, J^{\prime})$
is a forbidden pattern in $A$.
In addition, we have $|J^{\prime}| < |J|$ and $|I^{\prime}| < |I|$.
This contradicts the minimality of $(I,J)$.
This completes the proof.
\end{proof}

Here, we prove the contraposition of Theorem~\ref{theorem:necessary}.
That is, we prove that if $A$ has a forbidden pattern,
then $A$ is not universally connected. 

Let $(I,J) = (\{i_1, i_2, \ldots, i_{\lambda} \}, \{j_1, j_2, \ldots, j_{\lambda}\})$ 
be a forbidden pattern in $A$ that minimizes $|I| \, (=|J|)$,
where $\lambda := |I| \, (= |J|)$.
Notice that $\lambda \ge 2$ follows from Definition \ref{def:FP}.
Define $A^I = (a^I_{ij}) \in \mathbb{R}^{I \times [n]}$ 
as the submatrix of $A$ whose index set of rows is $I$.

First, we assume that $\lambda = 2$.
Lemma \ref{lem:FP->notEO} implies that $A^I$ does not have an elimination ordering.
Thus, Lemma \ref{lem:m<3_and_n<2_not_EO_->_not_UC}
implies that
there exists a vector $b^I \in \mathbb{R}^{[2]}$ such that $G(A^I,b^I)$ is not connected.
Define the vector $b \in \mathbb{R}^{[m]}$ 
by 
\begin{equation}
      b_k := \left\{
\begin{aligned}
&b^I_k && (k \in I) \\
&- dn\alpha_k && (k \notin I)
,\end{aligned}
\right.  
\label{eq:def_b_using_b^I}
\end{equation}
where we define 
$ \alpha_k := \max_{j \in [n]}|a_{k j}| $ for all integers $k \in [m]$.
It is not difficult to see that $G(A,b)$ is not connected.

Next, we assume that $\lambda \ge 3$.
First, we prove the following lemma.

\begin{lemma}\label{lem:notJ->elm_able}
    For all integers $i_s,i_t \in I$ and 
    $j \in [m] \setminus J$, we have 
    $a^I_{i_s j} a^I_{i_t j} \ge 0$.
\end{lemma}
\begin{proof}
We assume that there exist distinct integers $i^{\prime}, i^{\prime \prime} \in I$ and
$j \in [m] \setminus J$
such that $a^I_{i^{\prime} j} a^I_{i^{\prime \prime} j} < 0$.
Define $I^{\pm} := \{ (x,y) \in [\lambda] \times [\lambda] : a^I_{i_x j} a^I_{i_y j} < 0 \}$.
Then, fix $(s,t) \in I^{\pm}$ 
such that $|s- t| \le |x-y|$ for all pairs of integers $(x,y) \in I^{\pm}$.
Without loss of generality, we assume that $s>t$.
If there exists an integer $k \in \{t+1,\ldots,s-1\}$
such that $a_{i_k j} \neq 0$,
then at least one of $a^I_{i_s j} a^I_{i_k j} < 0$, 
$a^I_{i_t j} a^I_{i_k j} < 0$ holds.
This contradicts the minimality of $|s-t|$.
Hence, $a_{i_k j} = 0$ for all integers $k \in \{t+1,\ldots,s-1\}$.
Recall that we consider the case where $\lambda \ge 3$.

First, we consider the case where $t=1$ and $s=\lambda$.
We have $a^I_{i_s j_{\lambda}} a^I_{i_t j_{\lambda}} < 0$.
Define $I^{\prime} := \{i_s,i_t\}$
and $J^{\prime} := \{j_{\lambda}, j\}$.
Then $(I^{\prime}, J^{\prime})$
is a forbidden pattern in $A$.
This contradicts the minimality of $(I, J)$.

Next, we consider the case where at least one of $t \neq 1$, 
$s \neq \lambda$ holds.
Define $I^{\prime} := \{i_t, i_{t+1} \ldots, i_s\}$
and $J^{\prime} := \{j_t, j_{t+1}, \ldots, j_{s-1}, j\}$.
By the definition of $(I,J)$, 
for all integers $\ell \in \{t,\ldots,s-1\}$, 
$a_{i_{\ell} j_{\ell}} a_{i_{\ell+1} j_{\ell}} < 0$ and 
$a_{i j_{\ell}} = 0$,
where $i \in I^{\prime} \setminus \{ i_{\ell}, i_{\ell+1} \}$ is an arbitrary integer.
In the column $j$, $a^I_{i_s j} a^I_{i_t j} < 0$ and 
$a_{i_k j} = 0$ for all integers $k \in \{t+1,\ldots,s-1\}$.
Then $(I^{\prime}, J^{\prime})$
is a forbidden pattern in $A$.
This contradicts the minimality of $(I, J)$.
This completes the proof.
\end{proof}

Define $A^F \in \mathbb{R}^{I \times J}$ as the submatrix of $A^I$ whose index set of columns is $J$.
We reorder the columns and the rows of $A^F$ in such a way that
\(
A^F = (a^F_{s t}) = (a_{i_s j_t}).
\)
We define $a^F_{0 0} := a^F_{\lambda \lambda}$ and 
$a^F_{s 0} := a^F_{s \lambda}$ for all integers $s \in [\lambda]$.

We define the vector $b^F \in \mathbb{R}^{[\lambda]}$ by
\[
b^F_s := 
\left\{
\begin{aligned}
    &\sum_{k \in [\lambda]} a^F_{s k} d 
    \SF(a^{F}_{kk})
    && \left( \text{if } \max_{t \in [\lambda]}|a^{F}_{st}| \neq |a^{F}_{ss}| \right) \\
    & \sum_{k \in [\lambda]} a^F_{s k} \left( d
    \SF(a^{F}_{kk}) - \sgn(a^{F}_{kk}) \right)
    && \left( \text{if } \max_{t \in [\lambda]}|a^{F}_{st}| = |a^{F}_{ss}| \right) .
\end{aligned}
\right.
\]
Notice that $\max_{t \in [\lambda]}|a^{F}_{st}| \neq |a^{F}_{ss}|$ implies 
that 
$\max_{t \in [\lambda]}|a^{F}_{st}| = |a^{F}_{s ,s-1}|$
for all integers $s \in [\lambda]$.
Then we consider the integer linear system $(A^F,b^F)$.

We define the vectors $p,q \in D^{[\lambda]}$ by
\begin{align*}
    p_t &:= d \SF(a^{F}_{tt}), \qquad \\
    q_t &:= d \SF(a^{F}_{tt}) - \sgn(a^{F}_{tt})
\end{align*}
for all integers $t \in [\lambda]$.
We prove that $p,q \in R(A^F,b^F)$ and $p,q$ are not connected
on $G(A^F, b^F)$.

We consider $p$. 
We take an arbitrary integer $s \in [\lambda]$.
If $\max_{t \in [\lambda]}|a^{F}_{s t}| \neq |a^{F}_{s s}|$, 
then $\sum_{k \in [\lambda]} a^F_{s k} p_k - b^F_s = 0$.
We assume that $\max_{t \in [\lambda]}|a^{F}_{s t}| = |a^{F}_{s s}|$.
Then  
\begin{align*}
     \sum_{k \in [\lambda]} a^F_{s k} p_k - b^F_s 
    &= \sum_{k \in [\lambda]} a^F_{s k} 
    d \SF(a^{F}_{kk})
    - \sum_{k \in [\lambda]} a^F_{s k} \left( d
    \SF(a^{F}_{kk}) - \sgn(a^{F}_{kk}) \right)  \\
    &= \sum_{k \in [\lambda]} a^F_{s k} 
    \sgn(a^{F}_{kk}) \\
    &= a^F_{s s}\sgn(a^{F}_{ss}) + a^F_{s, s-1}\sgn(a^{F}_{s-1, s-1}) \\
    &= |a^F_{s s}| - |a^F_{s, s-1}| \ge 0 .
\end{align*}
This completes the proof of $p \in R(A^F,b^F)$.

We consider $q$. 
We take an arbitrary integer $s \in [\lambda]$.
If $\max_{t \in [\lambda]}|a^{F}_{s t}| = |a^{F}_{s s}|$, 
then $\sum_{k \in [\lambda]} a^F_{s k} q_k - b^F_s = 0$.
Assume that $\max_{t \in [\lambda]}|a^{F}_{s t}| \neq |a^{F}_{s s}|$.
Then  
\begin{align*}
     \sum_{k \in [\lambda]} a^F_{s k} q_k - b^F_s 
    &= \sum_{k \in [\lambda]} a^F_{s k} \left( d
    \SF(a^{F}_{kk}) - \sgn(a^{F}_{kk}) \right)
    - \sum_{k \in [\lambda]} a^F_{s k} 
    d \SF(a^{F}_{kk})  \\
    &= \sum_{k \in [\lambda]} - a^F_{s k} 
    \sgn(a^{F}_{kk}) \\
    &= - a^F_{ss}\sgn(a^{F}_{ss}) - a^F_{s, s-1}\sgn(a^{F}_{s-1, s-1}) \\
    &= - |a^F_{ss}| + |a^F_{s, s-1}| \ge 0.
\end{align*}
This completes the proof of $q \in R(A^F, b^F)$.

We define $Y$ as the set of vectors $y \in D^{[n]}$ such that $\dist(p,y) = 1$.
In other words, $Y$ is the set of 
neighborhood vertices of $p$ on $G(D^{[n]})$.
Then, we prove that $y \notin R(A^F,b^F)$ for all vectors $y \in Y$.

We arbitrarily take a vector $y \in Y$.
By the definition of $Y$,
there exists an integer $z \in [\lambda]$ such that $y_k = p_k$ for all integers $k \in [\lambda] \setminus \{z\} $ and $y_z = \xi$,
where $\xi$ is some integer in $D$ such that $\xi \neq p_{z}$.

Assume that $\max_{t \in [\lambda]}|a^{F}_{zt}| \neq |a^{F}_{zz}|$.
We have 
\begin{align*}
    \sum_{k \in [\lambda]} a^F_{z k} y_k - b_{z} 
    &= \sum_{k \in [\lambda]} a^F_{z k} y_k -  
    \sum_{k \in [\lambda]} a^F_{z k} 
    d \SF(a^{F}_{kk}) 
    = a^F_{z z} \left( \xi - d \SF(a^{F}_{zz}) \right).
\end{align*}

If $a^F_{z z} > 0$, then $p_z = d \SF(a^{F}_{zz}) = d$.
Since $\xi \neq p_z$, the inequality 
$0 \le \xi \le d-1$ is obtained. 
We have
\begin{align*}
    \sum_{k \in [\lambda]} a^F_{z k} y_k - b_{z} 
    = a^F_{z z}  \left( \xi - d \SF(a^{F}_{zz}) \right)
    \le a^F_{z z} ((d - 1) - d) 
    = - a^F_{z z} < 0.
\end{align*}

If $a^F_{z z} < 0$, then $p_z = d \SF(a^{F}_{zz}) = 0$.
Since $\xi \neq p_z$, the inequality 
$1 \le \xi \le d$ is obtained.  
We have
\begin{align*}
    \sum_{k \in [\lambda]} a^F_{z k} y_k - b_{z} 
    = a^F_{z z}  \left( \xi - d \SF(a^{F}_{zz}) \right)
    \le a^F_{z z} (1 - 0) 
    = a^F_{z z} < 0.
\end{align*}

Assume that $\max_{t \in [\lambda]}|a^{F}_{zt}| = |a^{F}_{zz}|$.
Then we have 
\begin{align*}
    \sum_{k \in [\lambda]} a^F_{z k} y_k - b_{z}
    &= \sum_{k \in [\lambda]} a^F_{z k} y_k 
    - \sum_{k \in [\lambda]} a^F_{z k} \left( d
    \SF(a^{F}_{kk}) - \sgn(a^{F}_{kk}) \right) \\
    &= a^F_{z z} \left(\xi - \left( d
    \SF(a^{F}_{zz}) - \sgn(a^{F}_{zz}) \right) \right) 
    + a^F_{z, z-1} \sgn(a^{F}_{z-1, z-1})  \\
    &= a^F_{z z} \left(\xi - \left( d
    \SF(a^{F}_{zz}) - \sgn(a^{F}_{zz}) \right) \right) 
    -| a^F_{z, z-1} |.
\end{align*}

If $a^F_{z z} > 0$, then $p_z = d \SF(a^{F}_{zz}) = d$.
Since $\xi \neq p_z$, the inequality 
$0 \le \xi \le d-1$ is obtained. 
We have
\begin{align*}
    \sum_{k \in [\lambda]} a^F_{{z} k} y_k - b_{z}  
    &= a^F_{z z} \left(\xi - \left( d
    \SF(a^{F}_{zz}) - \sgn(a^{F}_{zz}) \right) \right) 
    -| a^F_{z, z-1} | \\
    &\le a^F_{z z} ( d-1 - ( d - 1 ) )
    -| a^F_{z, z-1} | 
    = -| a^F_{z, z-1} |  < 0 .
\end{align*}

If $a^F_{z z} < 0$, then $p_z = d \SF(a^{F}_{zz}) = 0$.
Since $\xi \neq p_z$, the inequality 
$1 \le \xi \le d$ is obtained.  
We have
\begin{align*}
    \sum_{k \in [\lambda]} a^F_{{z} k} y_k - b_{z}  
    &= a^F_{z z} \left(\xi - \left( d
    \SF(a^{F}_{zz}) - \sgn(a^{F}_{zz}) \right) \right) 
    -| a^F_{z, z-1} | \\
    &\le a^F_{z z} ( 1 - ( 0 + 1 ) ) -| a^F_{z, z-1} | 
    = -| a^F_{z, z-1} |  < 0 .
\end{align*}

For all vectors $y \in Y$, we obtain $y \notin R(A^F, b^F)$.
Hence, any neighborhood vertex of $p$ on $G(D^{[n]})$ is not a feasible solution.
Thus, $G(A^F,b^F)$ is not connected.
That is, $A^F$ is not universally connected.
By Lemma~\ref{lem:notJ->elm_able},
for all integers $j \in [m] \setminus J$, 
$A^I$ can be eliminated at the column $j$.
Thus, by Lemma~\ref{lem:expansion_lemma},
$A^I$ is not universally connected
since $A^F$ is not universally connected.
There exists a vector $b^I \in \mathbb{R}^{[\lambda]}$ 
such that $G(A^I,b^I)$ is not connected.
By defining $b$ by \eqref{eq:def_b_using_b^I},
we can prove that $G(A,b)$ is not connected
in the same way as in the case where $\lambda = 2$.
This completes the proof of Theorem~\ref{theorem:necessary}.

\section{Proof of Theorem~\ref{theorem:sufficient}}

In this section, we prove Theorem~\ref{theorem:sufficient}. 
In this proof, we need the following lemmas.

\begin{lemma}[{\cite[Lemma 3]{KIMURA201667}}, {\cite[Lemma 5.3]{KIMURA202188}}]\label{lem:EO->UC}
    Let $A$ be a matrix in $\mathbb{R}^{[m]\times [n]}$.
    If $A$ has an elimination ordering,
    then $A$ is universally connected.
\end{lemma}

\begin{lemma}[{\cite[Proposition~12]{shigenobu2024connectednesssolutionsintegerlinear}}\footnote{
        This statement is slightly different from \cite[Proposition~12]{shigenobu2024connectednesssolutionsintegerlinear}.
        However, its proof 
        actually proves this statement.
        }]\label{lem:m2_notEO->FP}
    Let $A = (a_{ij})$ be a matrix in $\mathbb{R}^{[m] \times [2]}$.
    Suppose that $A$ does not have an elimination ordering.
    There exist distinct integers $i_1,i_2 \in [m]$ such that, 
    for all integers $j \in [2]$, 
    $\sgn(a_{i_1 j}) = -\sgn(a_{i_2 j}) \neq 0$.
\end{lemma}

First, we consider the case where at least one of (i) $m \le 3$,  
(ii) $n \le 2$ holds.
In order to prove this case, we prove the following Lemma~\ref{lem:m<3_and_n<2_not_EO->FP}.

\begin{lemma}\label{lem:m<3_and_n<2_not_EO->FP}
    Let $A = (a_{ij})$ be a matrix in $\mathbb{R}^{[m] \times [n]}$.
    Suppose that we have at least one of
    {\rm (i)} $m \le 3$, 
    {\rm (ii)} $n \le 2$.
    If $A$ does not have an elimination ordering,
    then $A$ has an a forbidden pattern.
\end{lemma}
\begin{proof}
If $m = 2$, then there exist distinct integers $j_1, j_2 \in [n]$ such that
$a_{1 j_1} a_{2 j_1} < 0$ and $a_{1 j_2} a_{2 j_2} < 0$
by Definition~\ref{def:elimination}.
Thus, $(\{1,2\}, \{j_1 , j_2\})$ is a forbidden pattern in $A$.

We consider the case where $m =3$.
Define $A^r$ as the submatrix of $A$ 
whose index set of columns is $\Delta$ defined by Algorithm~\ref{alg}.
We consider the case where
$|\Delta| = 2$.
Notice that $|\Delta|$ is the number of columns in $A^r$.
By Lemma~\ref{lem:m2_notEO->FP},
there exist distinct integers $i_1, i_2 \in [m]$
such that $a_{i_1 \delta} a_{i_2 \delta} < 0$ for all integers $\delta \in \Delta$.
Thus, $(\{i_1, i_2\}, \Delta)$ is a forbidden pattern in $A$.

Assume that $|\Delta| \ge 4$.
We fix arbitrary distinct integers $j_1, j_2, j_3, j_4 \in \Delta$. 
Notice that $A^r$ cannot be eliminated at the columns $j_1, j_2, j_3, j_4$.
Thus, there exist distinct integers $i^k_1, i^k_2 \in [m]$ such that
$a_{i^k_1 j_k} a_{i^k_2 j_k} < 0$ for all integers $k \in [4]$.
Since the number of $2$-combinations of $3$ elements is $3$,
there exist distinct integers $s,t \in \{1,2,3,4\}$ such that
$\{i^s_1, i^s_2\} = \{i^t_1, i^t_2\}$.
Thus, $(\{i^s_1, i^s_2\}, \{j_s, j_t\})$ is a forbidden pattern in $A$.

Assume that $|\Delta| = 3$
and $\Delta = \{j_1, j_2, j_3\}$.
First, we assume that there exists an integer $j^{\circ} \in \Delta$
such that $a_{i j^{\circ}} \neq 0$ for all integers $i \in [3]$.
Without loss of generality, we assume that $j^{\circ} = j_1$.
Then (i) there exist distinct integers $i_1, i_2, i_3 \in [3]$
such that $a_{i_1 j_1} a_{i_2 j_1} < 0$ and $a_{i_3 j_1} \neq 0$,
and (ii) at least one of $a_{i_1 j_1} a_{i_3 j_1} < 0$,
$a_{i_2 j_1} a_{i_3 j_1} < 0$ holds.
Without loss of generality, we assume that $a_{i_1 j_1} a_{i_3 j_1} < 0$.
There exist distinct integers $i^{\prime}_1, i^{\prime}_2 \in [3]$
such that $a_{i^{\prime}_1 j_2} a_{i^{\prime}_2 j_2} < 0$.
If $\{i^{\prime}_1, i^{\prime}_2\} = \{i_1, i_2\}$
(resp. $\{i^{\prime}_1, i^{\prime}_2\} = \{i_1, i_3\}$),
then $(\{i_1, i_2\}, \{j_1, j_2\})$ (resp. $(\{i_1, i_3\}, \{j_1, j_2\})$)
is a forbidden pattern in $A$.
Hence, we assume that $\{i^{\prime}_1, i^{\prime}_2\} = \{i_2, i_3\}$.
There exist distinct integers $i^{\prime \prime}_1, i^{\prime \prime}_2 \in [3]$
such that $a_{i^{\prime \prime}_1 j_3} a_{i^{\prime \prime}_2 j_3} < 0$.
Similarly, if $\{i^{\prime \prime}_1, i^{\prime \prime}_2\} = \{i_1, i_2\}$
(resp. $\{i^{\prime \prime}_1, i^{\prime \prime}_2\} = \{i_1, i_3\}$),
then $(\{i_1, i_2\}, \{j_1, j_3\})$ (resp. $(\{i_1, i_3\}, \{j_1, j_3\})$)
is a forbidden pattern in $A$.
Hence, we assume that $\{i^{\prime \prime}_1, i^{\prime \prime}_2\} = \{i_2, i_3\}$.
Thus, $(\{i_2, i_3\}, \{j_2, j_3\})$ is a forbidden pattern in $A$.

Next, we assume that, for all integers $j \in \Delta$,
there exists an integer $i \in [3]$
such that $a_{i j} = 0$.
Without loss of generality, we assume that $a_{1 j_1} a_{2 j_1} < 0$
and $a_{3 j_1} = 0$.
There exist distinct integers $i_1, i_2, i_3 \in [3]$ 
such that $a_{i_1 j_2} a_{i_2 j_2} < 0$ 
and $a_{i_3 j_2} = 0$.
If $\{i_1, i_2\} = \{1,2\}$,
then $(\{1, 2\}, \{j_1 , j_2\})$ is a forbidden pattern in $A$.
Hence, we can assume that $\{i_1, i_2\} \neq \{1,2\}$.
Without loss of generality, we assume that $\{i_1, i_2\} = \{2,3\}$
and $i_3 = 1$.
There exist distinct integers $i^{\prime}_1, i^{\prime}_2, i^{\prime}_3 \in [3]$ 
such that $a_{i^{\prime}_1 j_3} a_{i^{\prime}_2 j_3} < 0$ 
and $a_{i^{\prime}_3 j_3} = 0$.
If $\{i^{\prime}_1, i^{\prime}_2\} = \{1,2\}$ 
(resp. $\{i^{\prime}_1, i^{\prime}_2\} = \{2,3\}$),
then $(\{1, 2\}, \{j_1 , j_3\})$ 
(resp. $(\{2, 3\}, \{j_2, j_3\})$) is a forbidden pattern in $A$.
Hence, we assume that $\{i^{\prime}_1, i^{\prime}_2\} \neq \{1,2\}$ 
and $\{i^{\prime}_1, i^{\prime}_2\} \neq \{2,3\}$.
That is, we assume that $\{i^{\prime}_1, i^{\prime}_2\} = \{1,3\}$ 
and $i^{\prime}_3 = 2$.
In this case, 
$(\{1,2,3\}, \{j_1 , j_2, j_3\})$ is a forbidden pattern in $A$.

We consider the case where $n \le 2$.
Lemma~\ref{lem:m2_notEO->FP} implies that 
there exist distinct integers $i_1,i_2 \in [m]$ such that
$a_{i_1 1} a_{i_2 1} < 0$ and $a_{i_1 2} a_{i_2 2} < 0$.
Then $(\{i_1, i_2\}, \{1, 2\})$ is a forbidden pattern in $A$.
This completes the proof.
\end{proof}

In the above cases (i) and (ii), 
the contraposition of Lemma~\ref{lem:m<3_and_n<2_not_EO->FP} implies that
if $A$ does not have a forbidden pattern,
then $A$ has an elimination ordering.
By Lemma~\ref{lem:EO->UC}, $A$ is universally connected.

Next, we consider the case where (iii) $m = 4$ and $n = 3$.
By Lemma~\ref{lem:EO->UC}, if $A$ has an elimination ordering,
then Theorem~\ref{theorem:sufficient} holds.
Thus, we can assume that $A$ does not have an elimination ordering.

\begin{lemma}\label{3x4_EO_and_not_elm->FP}
    Let $A = (a_{ij})$ be a matrix in $\mathbb{R}^{[4]\times [3]}$.
    Suppose that $A$ does not have an elimination ordering.
    If there exists an integer $j \in [3]$ such that $A$ can be eliminated at the column $j$,
    then $A$ has a forbidden pattern.
\end{lemma}
\begin{proof}
Fix an integer $j \in [3]$ such that $A$ can be eliminated at the column $j$.
Define $A^{\prime} := \elm(A, \{j\})$.
By assumption, $A$ does not have an elimination ordering.
Hence, $A^{\prime}$ does not have an elimination ordering.
By Lemma~\ref{lem:m2_notEO->FP},
there exist distinct integers $i_1, i_2 \in [3]$ 
such that $a_{i_1 j^{\prime}} a_{i_2 j^{\prime}} < 0$ 
for all integers $j^{\prime} \in [3] \setminus \{j\}$.
Thus, $(\{i_1, i_2\}, [3] \setminus \{j\})$ is a forbidden pattern.
\end{proof}

Lemma~\ref{3x4_EO_and_not_elm->FP} implies that
if $A$ does not have an elimination ordering or a forbidden pattern,
then $A$ cannot be eliminated at any column of $A$.
The following lemma implies that $A$ has the unique sign pattern in this case.

\begin{lemma}\label{lem:unique_pattern_dicided}
    Let $A=(a_{ij})$ be a matrix in $\mathbb{R}^{[4] \times [3]}$.
    Suppose that $A$ cannot be eliminated at any column of $A$.
    If $A$ does not have a forbidden pattern,
    then $A$ can be transformed to a matrix having the following sign pattern
\begin{equation}
    \left(
    \begin{matrix}
        + & + & 0 \\
        + & - & 0\\
        - & 0 & + \\
        - & 0 & - 
    \end{matrix}
    \right) \label{eq:matrix_of_not_FP_and_not_EO}
\end{equation}
    by reordering the columns and the rows,
    where $+$ means a positive number and $-$ means a negative number.
\end{lemma}

We give the proof of Lemma~\ref{lem:unique_pattern_dicided} 
in Section~\ref{subsec:proof_of_lemma}.
The following lemma implies that a matrix having the 
sign pattern~\eqref{eq:matrix_of_not_FP_and_not_EO} is universally connected.
The proof of Lemma~\ref{lem:connected_with_unique_pattern_in_general} is 
the same as {\cite[Lemma 1]{shigenobuCOCOA2023}}.
For completeness, we give its proof. 
Thus, this completes the proof of Theorem~\ref{theorem:sufficient}.

\begin{lemma}\label{lem:connected_with_unique_pattern_in_general}
    Let $A$ be a matrix in $\mathbb{R}^{[4] \times [3]}$ having the following form
    \begin{equation}
    A = 
    \left(
    \begin{matrix}
        p_{1 1} & p_{1 2} & 0 \\
        p_{2 1} & q_{2 2} & 0 \\
        q_{3 1} & 0 & p_{3 3} \\
        q_{4 1} & 0 & q_{4 3}
    \end{matrix}
    \right), \label{eq:matrix_of_not_FP _and_not_EO_with_unique_pattern}
    \end{equation}
where $p_{1 1}, p_{1 2}, p_{2 1}, p_{3 3} > 0$ and $q_{2 2}, q_{3 1}, q_{4 1}, q_{4 3} < 0$.
Then $A$ is universally connected.
\end{lemma}
\begin{proof}
Fix an arbitrary vector $b \in \mathbb{R}^{[4]}$.
If $|R(A,b)| < 2$,
then $G(A,b)$ is clearly connected.
Assume that $R(A,b)$ contains at least two solutions.
We take arbitrary vectors $s,t \in R(A,b)$.
Without loss of generality, we assume that $s_1 \ge t_1$.
Let $P$ be the path from $s$ to $t$ defined by
\[
    s = 
    \begin{pmatrix}
        s_1 \\
        s_2 \\
        s_3
    \end{pmatrix}
    \to
    u^1 = 
    \begin{pmatrix}
        s_1 \\
        t_2 \\
        s_3
    \end{pmatrix}
    \to
    u^2 = 
    \begin{pmatrix}
        t_1 \\
        t_2 \\
        s_3
    \end{pmatrix}
    \to
    t = 
    \begin{pmatrix}
        t_1 \\
        t_2 \\
        t_3
    \end{pmatrix}.
\]
Then we have $u^1 \in R(A,b)$
because 
\begin{align*}
    p_{1 1} s_1 + p_{1 2} t_2  \ge p_{1 1} t_1 + p_{1 2} t_2  \ge b_1, ~~
    &q_{3 1} s_1 + p_{3 3} s_3  \ge b_3, \\
    p_{2 1} s_1 + q_{2 2} t_2  \ge p_{2 1} t_1 + q_{2 2} t_2  \ge b_2, ~~
    &q_{4 1} s_1 + q_{4 3} s_3  \ge b_4.
\end{align*}
Furthermore, we have $u^2 \in R(A,b)$
because 
\begin{align*}
    p_{1 1} t_1 + p_{1 2} t_2  \ge b_1, ~~
    &q_{3 1} t_1 + p_{3 3} s_3  \ge q_{3 1} s_1 + p_{3 3} s_3 \ge b_3, \\
    p_{2 1} t_1 + q_{2 2} t_2  \ge b_2, ~~
    &q_{4 1} t_1 + q_{4 3} s_3  \ge q_{4 1} s_1 + q_{4 3} s_3 \ge b_4.
\end{align*}
These imply that  
every vertex of $P$ is contained in 
$R(A,b)$, and $G(A,b)$ is connected.
This completes the proof.
\end{proof}

\subsection{Proof of Lemma~\ref{lem:unique_pattern_dicided}}
\label{subsec:proof_of_lemma}
Throughout this subsection, 
for all integers $j \in [3]$, 
let $\otimes_j$, $\odot_j$ be symbols 
satisfying the condition that 
$\{ \otimes_j, \odot_j \} = \{+, -\}$.
Notice that,
for all distinct integers $j_1, j_2 \in [3]$ and all symbols $\square \in \{\otimes, \odot\}$,
there is a possibility that $\square_{j_1} \neq \square_{j_2}$.
In what follows, for all integers $i \in [4]$ and $j \in [3]$ 
and all symbols $\square_j \in \{\otimes_j, \odot_j\}$, if we write $a_{ij} = \square_j$,
then this means that the sign of $a_{ij}$ is $\square_j$. 
Similarly, $a_{ij} \neq \square_j$ means that the sign of $a_{ij}$ is not $\square_j$.

Since $A$ cannot be eliminated at any column,
there exist integers $i_1 \in [4]$ and $j \in \{2,3\}$
such that $a_{i_1 1} = \otimes_{1}$ and $a_{i_1 j} \neq 0$.
Without loss of generality, we can assume that 
$a_{i_1 j} = \otimes_{j}$.
Furthermore, without loss of generality, we can assume that $i_1 = 1$ and $j = 2$.
Since $A$ cannot be eliminated at any column,
there exists an integer $i_2 \in \{ 2,3,4 \}$
such that $a_{i_2 1} = \odot_{1}$.
Without loss of generality, we can assume that $i_2 = 2$.
Similarly, there exists an integer $i^{\prime}_2 \in \{ 2,3,4 \}$
such that $a_{i^{\prime}_2 2} = \odot_{2}$.
If $i^{\prime}_2 = 2$, then $(\{1,2\}, \{1,2\})$ is a forbidden pattern in $A$.
Hence, by assumption, $i^{\prime}_2 \neq 2$.
Thus, without loss of generality, we can assume that $i^{\prime}_2 = 3$,
i.e., $A$ has the following form
\begin{equation*}
\begin{pmatrix}
    \otimes_1 & \otimes_2 & \quad \\
    \odot_1   &           &       \\
    ~         & \odot_2   &       \\
    ~         &           & 
\end{pmatrix}, 
\end{equation*}
where blanks mean that the sign patterns are not determined.
First, we consider $a_{3 1}$.
Notice that if $a_{3 1} = \odot_{1}$, 
then $(\{1,3\}, \{1,2\})$ is a forbidden pattern.
First, we perform a case distinction on whether 
$a_{3 1} = \otimes_{1}$ or $a_{3 1} = 0$.

\textbf{Case 1.}
Assume that $a_{3 1} = \otimes_{1}$.
If $a_{2 2} = \otimes_{2}$,
then $( \{2,3\}, \{1,2\})$ is a forbidden pattern.
If $a_{2 2} = \odot_{2}$,
then $( \{1,2\}, \{1,2\})$ is a forbidden pattern.
Hence, $a_{2 2} = 0$.
\[
\begin{pmatrix}
    \otimes_1 & \otimes_2 & \quad \\
    \odot_1   & 0         &       \\
    \otimes_1 & \odot_2   &       \\
              &           & 
\end{pmatrix}.
\]
Next, we perform a case distinction on whether $a_{2 3} = 0$ or $a_{2 3} \neq 0$.

\textbf{Case 1.1.}
Assume that $a_{2 3} = 0$.
If $a_{4 1} \neq \odot_{1}$,
then $A$ can be eliminated at the column $1$.
Hence, $a_{4 1} = \odot_{1}$.
\[
\begin{pmatrix}
    \otimes_1 & \otimes_2 & \quad \\
    \odot_1   & 0                & 0        \\
    \otimes_1 & \odot_2   &  \\
    \odot_1   &           & 
\end{pmatrix}.
\]
If $a_{4 2} = 0$ and $a_{4 3} = 0$,
then $A$ can be eliminated at the column $1$.
Hence, we have at least one of $a_{4 2} \neq 0$,  $a_{4 3} \neq 0$.
If $a_{4 2} = \otimes_{2}$,
then $( \{3,4\}, \{1,2\})$ is a forbidden pattern.
If $a_{4 2} = \odot_{2}$,
then $( \{1,4\}, \{1,2\})$ is a forbidden pattern.
Hence, $a_{4 2} = 0$.
Thus, we can assume that $a_{4 3} = \otimes_{3}$.
\[
    \begin{pmatrix}
    \otimes_1 & \otimes_2 &   \\
    \odot_1   & 0         & 0               \\
    \otimes_1 & \odot_2   &   \\
    \odot_1   & 0         & \otimes_3
    \end{pmatrix}. 
\]
If $a_{1 3} = \odot_{3}$,
then $( \{1,4\}, \{1,3\})$ is a forbidden pattern.
Furthermore, if $a_{3 3} = \odot_{3}$,
then $( \{3,4\}, \{1,3\})$ is a forbidden pattern.
These mean that $A$ can be eliminated at the column $3$.
Thus, Case 1.1 does not happen.

\textbf{Case 1.2.}
Assume that $a_{2 3} \neq 0$.
Without loss of generality, 
we can assume that $a_{2 3} = \otimes_{3}$.
\[
    \begin{pmatrix}
    \otimes_1 & \otimes_2 &   \\
    \odot_1   & 0         & \otimes_3 \\
    \otimes_1 & \odot_2   &   \\
      &   & 
    \end{pmatrix}. 
\]
If $a_{1 3} = \odot_{3}$,
then $( \{1,2\}, \{1,3\})$ is a forbidden pattern.
If $a_{3 3} = \odot_{3}$,
then $( \{2,3\}, \{1,3\})$ is a forbidden pattern.
Hence, $a_{4 3} = \odot_{3}$.
\[
    \begin{pmatrix}
    \otimes_1 & \otimes_2 &   \\
    \odot_1   & 0         & \otimes_3 \\
    \otimes_1 & \odot_2   &   \\
      &   & \odot_3  
    \end{pmatrix}. 
\]
If $a_{4 1} = \otimes_{1}$,
then $( \{2,4\}, \{1,3\})$ is a forbidden pattern.
Next, we perform a case distinction on whether $a_{4 1} = 0$ or $a_{4 1} = \odot_{1}$.

\textbf{Case 1.2.1.}
Assume that $a_{4 1} = 0$.
\[
    \begin{pmatrix}
    \otimes_1 & \otimes_2 &   \\
    \odot_1   & 0         & \otimes_3 \\
    \otimes_1 & \odot_2   &   \\
    0        &   & \odot_3  
    \end{pmatrix}. 
\]
We perform a case distinction on whether
$a_{4 2} = \otimes_{2}$, or $a_{4 2} = \odot_{2}$, or $a_{4 2} = 0$.
Assume that $a_{4 2} = \otimes_{2}$.
\[
    \begin{pmatrix}
    \otimes_1 & \otimes_2 &   \\
    \odot_1   & 0         & \otimes_3 \\
    \otimes_1 & \odot_2   &   \\
    0         & \otimes_2 & \odot_3  
    \end{pmatrix}. 
\]
If $a_{3 3} = \otimes_{3}$,
then $( \{3,4\}, \{2,3\})$ is a forbidden pattern.
If $a_{3 3} = \odot_{3}$,
then $( \{2,3\}, \{1,3\})$ is a forbidden pattern.
If $a_{3 3} = 0$,
then $( \{2,3,4\}, \{1,2,3\})$ is a forbidden pattern.
Hence, in Case 1.2.1, $a_{4 2} \neq \otimes_{2}$.
Assume that $a_{4 2} = \odot_{2}$.
\[
    \begin{pmatrix}
    \otimes_1 & \otimes_2 &   \\
    \odot_1   & 0         & \otimes_3 \\
    \otimes_1 & \odot_2   &   \\
    0         & \odot_2   & \odot_3  
    \end{pmatrix}. 
\]
If $a_{1 3} = \otimes_{3}$,
then $( \{1,4\}, \{2,3\})$ is a forbidden pattern.
If $a_{1 3} = \odot_{3}$,
then $( \{1,2\}, \{1,3\})$ is a forbidden pattern.
If $a_{1 3} = 0$,
then $( \{1,2,4\}, \{1,3,2\})$ is a forbidden pattern.
Hence, in Case 1.2.1, $a_{4 2} \neq \odot_{2}$.
Assume that $a_{4 2} = 0$.
\[
    \begin{pmatrix}
    \otimes_1 & \otimes_2 &   \\
    \odot_1   & 0         & \otimes_3 \\
    \otimes_1 & \odot_2   &   \\
    0         & 0         & \odot_3  
    \end{pmatrix}. 
\]
If $a_{1 3} \neq \odot_{3}$ and $a_{3 3} \neq \odot_{3}$,
then $A$ can be eliminated at the column $3$.
Hence, we have at least one of $a_{1 3} = \odot_{3}$,  $a_{3 3} = \odot_{3}$.
If $a_{1 3} = \odot_{3}$,
then $( \{1,2\}, \{1,3\})$ is a forbidden pattern.
If $a_{3 3} = \odot_{3}$,
then $( \{2,3\}, \{1,3\})$ is a forbidden pattern.
Hence, in Case 1.2.1, $a_{4 2} \neq 0$.
Thus, Case 1.2.1 does not happen.

\textbf{Case 1.2.2.}
Assume that $a_{4 1} = \odot_{1}$.
\[
    \begin{pmatrix}
    \otimes_1 & \otimes_2 &   \\
    \odot_1   & 0        & \otimes_3 \\
    \otimes_1 & \odot_2   &   \\
    \odot_1   &   & \odot_3  
    \end{pmatrix}. 
\]
If $a_{4 2} = \otimes_{2}$,
then $( \{3,4\}, \{1,2\})$ is a forbidden pattern.
Furthermore, if $a_{4 2} = \odot_{2}$,
then $( \{1,4\}, \{1,2\})$ is a forbidden pattern.
Hence, $a_{4 2} = 0$.
\[
    \begin{pmatrix}
    \otimes_1 & \otimes_2 &   \\
    \odot_1   & 0         & \otimes_3 \\
    \otimes_1 & \odot_2   &   \\
    \odot_1   & 0         & \odot_3  
    \end{pmatrix}. 
\]
If $a_{1 3} = \otimes_{3}$,
then $( \{1,4\}, \{1,3\})$ is a forbidden pattern.
If $a_{1 3} = \odot_{3}$,
then $( \{1,2\}, \{1,3\})$ is a forbidden pattern.
Hence, $a_{1 3} = 0$.
If $a_{3 3} = \otimes_{3}$,
then $( \{3,4\}, \{1,3\})$ is a forbidden pattern.
If $a_{3 3} = \odot_{3}$,
then $( \{2,3\}, \{1,3\})$ is a forbidden pattern.
Hence, $a_{3 3} = 0$.
\[
    \begin{pmatrix}
    \otimes_1 & \otimes_2 & 0        \\
    \odot_1   & 0         & \otimes_3 \\
    \otimes_1 & \odot_2   & 0        \\
    \odot_1   & 0         & \odot_3  
    \end{pmatrix}. 
\]
No matter how $\otimes_j, \, \odot_j$ are replaced with $+, \, -$
for all integers $j \in [3]$,
this matrix can be transformed to the matrix~\eqref{eq:matrix_of_not_FP_and_not_EO}
by reordering the columns and the rows.

\textbf{Case 2.}
Assume that $a_{3 1} = 0$.
\[
    \begin{pmatrix}
    \otimes_1 & \otimes_2 & \quad \\
    \odot_1   &   &  \\
    0         & \odot_2   &  \\
      &   & 
    \end{pmatrix}.
\]
Next, we perform a case distinction on whether $a_{3 3} = 0$ or $a_{3 3} \neq 0$.

\textbf{Case 2.1.}
Assume that $a_{3 3} = 0$.
If $a_{2 2} = \odot_{2}$,
then $( \{1,2\}, \{1,2\})$ is a forbidden pattern.
Hence, $a_{4 2} = \odot_{2}$.
\[
    \begin{pmatrix}
    \otimes_1 & \otimes_2 & \quad \\
    \odot_1   &           &  \\
    0         & \odot_2   & 0        \\
    ~         & \odot_2   & 
    \end{pmatrix}.
\]
If $a_{4 1} = \odot_{1}$,
then $( \{1,4\}, \{1,2\})$ is a forbidden pattern.
Next, we perform a case distinction on whether $a_{4 1} = \otimes_{1}$ or $a_{4 1} = 0$.

\textbf{Case 2.1.1.}
Assume that $a_{4 1} = \otimes_{1}$.
\[
    \begin{pmatrix}
    \otimes_1 & \otimes_2 & \quad \\
    \odot_1   &           &  \\
    0         & \odot_2   & 0        \\
    \otimes_1 & \odot_2   & 
    \end{pmatrix}.
\]
If $a_{2 2} = \otimes_{2}$,
then $( \{2,4\}, \{1,2\})$ is a forbidden pattern.
If $a_{2 2} = \odot_{2}$,
then $( \{1,2\}, \{1,2\})$ is a forbidden pattern.
Hence, $a_{2 2} = 0$.
If $a_{2 3} = 0$,
then $A$ can be eliminated at the column $1$.
Hence, $a_{2 3} \neq 0$.
Without loss of generality,
we can assume that $a_{2 3} = \otimes_{3}$.
\[
    \begin{pmatrix}
    \otimes_1 & \otimes_2 &   \\
    \odot_1   & 0         & \otimes_3 \\
    0         & \odot_2   & 0        \\
    \otimes_1 & \odot_2   & 
    \end{pmatrix}.
\]
If $a_{1 3} = \odot_{3}$,
then $( \{1,2\}, \{1,3\})$ is a forbidden pattern.
If $a_{4 3} = \odot_{3}$,
then $( \{2,4\}, \{1,3\})$ is a forbidden pattern.
Hence, Case 2.1.1 does not happen.

\textbf{Case 2.1.2.}
Assume that $a_{4 1} = 0$.
If $a_{4 3} = 0$,
then $A$ can be eliminated at the column $2$.
Hence, without loss of generality, we can assume that $a_{4 3} = \otimes_{3}$.
\[
    \begin{pmatrix}
    \otimes_1 & \otimes_2 &   \\
    \odot_1   &           &   \\
    0         & \odot_2   & 0        \\
    0         & \odot_2   & \otimes_3
    \end{pmatrix}.
\]
If $a_{1 3} = \odot_{3}$,
then $( \{1,4\}, \{2,3\})$ is a forbidden pattern.
Hence, $a_{1 3} \neq \odot_{3}$.
If $a_{2 3} \neq \odot_{3}$,
then $A$ can be eliminated at the column $3$.
Hence, $a_{2 3} = \odot_{3}$.
\[
    \begin{pmatrix}
    \otimes_1 & \otimes_2 &   \\
    \odot_1   &           & \odot_3   \\
    0         & \odot_2   & 0        \\
    0         & \odot_2   & \otimes_3
    \end{pmatrix}.
\]
If $a_{2 2} = \otimes_{2}$,
then $( \{2,4\}, \{2,3\})$ is a forbidden pattern.
If $a_{2 2} = \odot_{2}$,
then $( \{1,2\}, \{1,2\})$ is a forbidden pattern.
Hence, in Case 2.1.2, $a_{2 2} = 0$.
If $a_{1 3} = \otimes_{3}$,
then $( \{1,2\}, \{1,3\})$ is a forbidden pattern.
If $a_{1 3} = \odot_{3}$,
then $( \{1,4\}, \{2,3\})$ is a forbidden pattern.
If $a_{1 3} = 0$,
then $( \{1,2,4\}, \{1,3,2\})$ is a forbidden pattern.
Hence, in Case 2.1, $a_{4 1} \neq 0$.
Thus, Case 2.1 does not happen.

\textbf{Case 2.2.}
Assume that $a_{3 3} \neq 0$.
Without loss of generality,
we can assume that $a_{3 3} = \otimes_{3}$.
\[
    \begin{pmatrix}
    \otimes_1 & \otimes_2 &   \\
    \odot_1   &           &   \\
    0         & \odot_2   & \otimes_3\\
    ~         &           &  
    \end{pmatrix}.
\]
If $a_{2 2} = \odot_{2}$,
then $( \{1,2\}, \{1,2\})$ is a forbidden pattern.
Next, we perform a case distinction on whether $a_{2 2} = \otimes_{2}$ or $a_{2 2} = 0$.

\textbf{Case 2.2.1.}
Assume that $a_{2 2} = \otimes_{2}$.
\[
    \begin{pmatrix}
    \otimes_1 & \otimes_2 &   \\
    \odot_1   & \otimes_2 &   \\
    0         & \odot_2   & \otimes_3 \\
      &   &  
    \end{pmatrix}.
\] 
If $a_{1 3} = \odot_{3}$,
then $( \{1,3\}, \{2,3\})$ is a forbidden pattern.
If $a_{2 3} = \odot_{3}$,
then $( \{2,3\}, \{2,3\})$ is a forbidden pattern.
Hence, in Case 2.2.1, $a_{4 3} = \odot_{3}$.
\[
    \begin{pmatrix}
    \otimes_1 & \otimes_2 &   \\
    \odot_1   & \otimes_2 &   \\
    0         & \odot_2   & \otimes_3 \\
      &                   & \odot_3  
    \end{pmatrix}.
\]
If $a_{4 2} = \otimes_{2}$,
then $( \{3,4\}, \{2,3\})$ is a forbidden pattern.
Next, we perform a case distinction on whether $a_{4 2} = \odot_{2}$ or $a_{4 2} = 0$.

\textbf{Case 2.2.1.1}
Assume that $a_{4 2} = \odot_{2}$.
\[
    \begin{pmatrix}
    \otimes_1 & \otimes_2 &   \\
    \odot_1   & \otimes_2 &   \\
    0         & \odot_2   & \otimes_3 \\
    ~         & \odot_2   & \odot_3  
    \end{pmatrix}.
\]
If $a_{4 1} = \otimes_{1}$,
then $( \{2,4\}, \{1,2\})$ is a forbidden pattern.
If $a_{4 1} = \odot_{1}$,
then $( \{1,4\}, \{1,2\})$ is a forbidden pattern.
Hence, in Case 2.2.1.1, $a_{4 1} = 0$.
If $a_{1 3} = \otimes_{3}$,
then $( \{1,4\}, \{2,3\})$ is a forbidden pattern.
If $a_{1 3} = \odot_{3}$,
then $( \{1,3\}, \{2,3\})$ is a forbidden pattern.
Hence, in Case 2.2.1.1, $a_{1 3} = 0$.
If $a_{2 3} = \otimes_{3}$,
then $( \{2,4\}, \{2,3\})$ is a forbidden pattern.
If $a_{2 3} = \odot_{3}$,
then $( \{2,3\}, \{2,3\})$ is a forbidden pattern.
Hence, in Case 2.2.1.1, $a_{2 3} = 0$.
\[
    \begin{pmatrix}
    \otimes_1 & \otimes_2 & 0         \\
    \odot_1   & \otimes_2 & 0         \\
    0         & \odot_2   & \otimes_3 \\
    0         & \odot_2   & \odot_3  
    \end{pmatrix}.
\]
No matter how $\otimes_j, \, \odot_j$ are replaced with $+, \, -$
for all integers $j \in [3]$,
this matrix can be transformed to the matrix~\eqref{eq:matrix_of_not_FP_and_not_EO}
by reordering the columns and the rows.

\textbf{Case 2.2.1.2.}
Assume that $a_{4 2} = 0$.
\[
    \begin{pmatrix}
    \otimes_1 & \otimes_2 &   \\
    \odot_1   & \otimes_2 &   \\
    0         & \odot_2   & \otimes_3 \\
    ~         & 0         & \odot_3  
    \end{pmatrix}.
\]
We perform a case distinction on whether
$a_{4 1} = \otimes_{1}$, or $a_{4 1} = \odot_{1}$, or $a_{4 1} = 0$.

\textbf{Case 2.2.1.2.1.}
Assume that $a_{4 1} = \otimes_{1}$.
\[
    \begin{pmatrix}
    \otimes_1 & \otimes_2 &   \\
    \odot_1   & \otimes_2 &   \\
    0         & \odot_2   & \otimes_3 \\
    \otimes_1 & 0         & \odot_3  
    \end{pmatrix}.
\]
If $a_{2 3} = \otimes_{3}$,
then $( \{2,4\}, \{1,3\})$ is a forbidden pattern.
If $a_{2 3} = \odot_{3}$,
then $( \{2,3\}, \{2,3\})$ is a forbidden pattern.
If $a_{2 3} = 0$,
then $( \{2,3,4\}, \{1,2,3\})$ is a forbidden pattern.
Hence, Case 2.2.1.2.1 does not happen.

\textbf{Case 2.2.1.2.2.}
Assume that $a_{4 1} = \odot_{1}$.
\[
    \begin{pmatrix}
    \otimes_1 & \otimes_2 &   \\
    \odot_1   & \otimes_2 &   \\
    0         & \odot_2   & \otimes_3 \\
    \odot_1   & 0         & \odot_3  
    \end{pmatrix}.
\]
If $a_{1 3} = \otimes_{3}$,
then $( \{1,4\}, \{1,3\})$ is a forbidden pattern.
If $a_{1 3} = \odot_{3}$,
then $( \{1,3\}, \{2,3\})$ is a forbidden pattern.
If $a_{1 3} = 0$,
then $( \{1,3,4\}, \{1,2,3\})$ is a forbidden pattern.
Hence, Case 2.2.1.2.2 does not happen.

\textbf{Case 2.2.1.2.3.}
Assume that $a_{4 1} = 0$.
\[
    \begin{pmatrix}
    \otimes_1 & \otimes_2 &   \\
    \odot_1   & \otimes_2 &   \\
    0         & \odot_2   & \otimes_3 \\
    0         & 0         & \odot_3  
    \end{pmatrix}.
\]
If $a_{1 3} = \odot_{3}$,
then $( \{1,3\}, \{2,3\})$ is a forbidden pattern.
If $a_{2 3} = \odot_{3}$,
then $( \{2,3\}, \{2,3\})$ is a forbidden pattern.
This implies that $A$ can be eliminated at the column $3$.
Case 2.2.1.2.3
does not happen.
Thus, Case 2.2.1.2 does not happen.

\textbf{Case 2.2.2.}
Assume that $a_{2 2} = 0$.
\[
    \begin{pmatrix}
    \otimes_1 & \otimes_2 &   \\
    \odot_1   & 0         &   \\
    0         & \odot_2   & \otimes_3 \\
    ~         &   &  
    \end{pmatrix}.
\]
We perform a case distinction on whether
$a_{2 3} = \otimes_{3}$, or $a_{2 3} = \odot_{3}$, or $a_{2 3} = 0$.

\textbf{Case 2.2.2.1.}
Assume that $a_{2 3} = \otimes_{3}$.
\[
    \begin{pmatrix}
    \otimes_1 & \otimes_2 &   \\
    \odot_1   & 0         & \otimes_3 \\
    0         & \odot_2   & \otimes_3 \\
      &   &  
    \end{pmatrix}.
\]
If $a_{1 3} = \odot_{3}$,
then $( \{1,3\}, \{2,3\})$ is a forbidden pattern.
Hence, $a_{4 3} = \odot_{3}$.
\[
    \begin{pmatrix}
    \otimes_1 & \otimes_2 &   \\
    \odot_1   & 0         & \otimes_3 \\
    0         & \odot_2   & \otimes_3 \\
    ~         &           & \odot_3  
    \end{pmatrix}.
\]
If $a_{4 2} = \otimes_{2}$,
then $( \{3,4\}, \{2,3\})$ is a forbidden pattern.
Next, we perform a case distinction on whether $a_{4 2} = \odot_{2}$ or $a_{4 2} = 0$.

\textbf{Case 2.2.2.1.1.}
Assume that $a_{4 2} = \odot_{2}$.
\[
    \begin{pmatrix}
    \otimes_1 & \otimes_2 &   \\
    \odot_1   & 0         & \otimes_3 \\
    0         & \odot_2   & \otimes_3 \\
    ~         & \odot_2   & \odot_3  
    \end{pmatrix}.
\]
If $a_{4 1} = \otimes_{1}$,
then $( \{2,4\}, \{1,3\})$ is a forbidden pattern.
If $a_{4 1} = \odot_{1}$,
then $( \{1,4\}, \{1,2\})$ is a forbidden pattern.
Hence, $a_{4 1} = 0$.
\[
    \begin{pmatrix}
    \otimes_1 & \otimes_2 &   \\
    \odot_1   & 0         & \otimes_3 \\
    0         & \odot_2   & \otimes_3 \\
    0         & \odot_2   & \odot_3  
    \end{pmatrix}.
\] 
If $a_{1 3} = \otimes_{3}$,
then $( \{1,4\}, \{2,3\})$ is a forbidden pattern.
If $a_{1 3} = \odot_{3}$,
then $( \{1,3\}, \{2,3\})$ is a forbidden pattern.
If $a_{1 3} = 0$,
then $( \{1,2,4\}, \{1,2,3\})$ is a forbidden pattern.
Hence, Case 2.2.2.1.1 does not happen.

\textbf{Case 2.2.2.1.2.}
Assume that $a_{4 2} = 0$.
\[
    \begin{pmatrix}
    \otimes_1 & \otimes_2 &   \\
    \odot_1   & 0         & \otimes_3 \\
    0         & \odot_2   & \otimes_3 \\
    ~         & 0         & \odot_3  
    \end{pmatrix}.
\]
If $a_{4 1} = \otimes_{1}$,
then $( \{2,4\}, \{1,3\})$ is a forbidden pattern.
Assume that $a_{4 1} = \odot_{1}$.
\[
    \begin{pmatrix}
    \otimes_1 & \otimes_2 &   \\
    \odot_1   & 0         & \otimes_3 \\
    0         & \odot_2   & \otimes_3 \\
    \odot_1   & 0         & \odot_3  
    \end{pmatrix}.
\]
If $a_{1 3} = \otimes_{3}$,
then $( \{1,4\}, \{1,3\})$ is a forbidden pattern.
If $a_{1 3} = \odot_{3}$,
then $( \{1,3\}, \{2,3\})$ is a forbidden pattern.
If $a_{1 3} = 0$,
then $( \{1,3,4\}, \{1,2,3\})$ is a forbidden pattern.
Hence, in Case 2.2.2.1.2, $a_{4 1} \neq \odot_{1}$.
Assume that $a_{4 1} = 0$.
\[
    \begin{pmatrix}
    \otimes_1 & \otimes_2 &   \\
    \odot_1   & 0         & \otimes_3 \\
    0         & \odot_2   & \otimes_3 \\
    0         & 0         & \odot_3  
    \end{pmatrix}.
\] 
If $a_{1 3} = \odot_{3}$,
then $( \{1,2\}, \{1,3\})$ is a forbidden pattern.
Hence, $a_{1 3} \neq \odot_{3}$.
This implies that $A$ can be eliminated at the column $3$.
Hence, $a_{4 1} \neq 0$.
Hence, Case 2.2.2.1.2 does not happen.
Thus, Case 2.2.2.1 does not happen.

\textbf{Case 2.2.2.2.}
Assume that $a_{2 3} = \odot_{3}$.
\[
    \begin{pmatrix}
    \otimes_1 & \otimes_2 &   \\
    \odot_1   & 0         & \odot_3   \\
    0         & \odot_2   & \otimes_3 \\
      &   &  
    \end{pmatrix}.
\]
If $a_{1 3} = \otimes_{3}$,
then $( \{1,2\}, \{1,3\})$ is a forbidden pattern.
If $a_{1 3} = \odot_{3}$,
then $( \{1,3\}, \{2,3\})$ is a forbidden pattern.
If $a_{1 3} = 0$,
then $( \{1,2,3\}, \{1,2,3\})$ is a forbidden pattern.
Hence, Case 2.2.2.2 does not happen.

\textbf{Case 2.2.2.3.}
Assume that $a_{2 3} = 0$.
\[
    \begin{pmatrix}
    \otimes_1 & \otimes_2 &   \\
    \odot_1   & 0         & 0        \\
    0         & \odot_2   & \otimes_3 \\
    ~         &           &  
    \end{pmatrix}.
\]
If $a_{4 1} \neq \odot_{1}$,
then $A$ can be eliminated at the column $1$.
Hence, $a_{4 1} = \odot_{1}$.
\[
    \begin{pmatrix}
    \otimes_1 & \otimes_2 &   \\
    \odot_1   & 0         & 0        \\
    0         & \odot_2   & \otimes_3 \\
    \odot_1   &   &  
    \end{pmatrix}.
\]
If $a_{1 3} \neq \odot_{3}$ and $a_{4 3} \neq \odot_{3}$,
then $A$ can be eliminated at the column $3$.
Hence, we have at least one of $a_{1 3} = \odot_{3}$,  $a_{4 3} = \odot_{3}$.
If $a_{1 3} = \odot_{3}$,
then $( \{1,3\}, \{2,3\})$ is a forbidden pattern.
Hence, $a_{4 3} = \odot_{3}$.
\[
    \begin{pmatrix}
    \otimes_1 & \otimes_2 &   \\
    \odot_1   & 0         & 0        \\
    0         & \odot_2   & \otimes_3 \\
    \odot_1   &           & \odot_3  
    \end{pmatrix}.
\] 
If $a_{4 2} = \otimes_{2}$,
then $( \{3,4\}, \{2,3\})$ is a forbidden pattern.
If $a_{4 2} = \odot_{2}$,
then $( \{1,4\}, \{1,2\})$ is a forbidden pattern.
Hence, $a_{4 2} = 0$.
\[
    \begin{pmatrix}
    \otimes_1 & \otimes_2 &   \\
    \odot_1   & 0         & 0        \\
    0         & \odot_2   & \otimes_3 \\
    \odot_1   & 0         & \odot_3  
    \end{pmatrix}.
\]
If $a_{1 3} = \otimes_{3}$,
then $( \{1,4\}, \{1,3\})$ is a forbidden pattern.
If $a_{1 3} = \odot_{3}$,
then $( \{1,3\}, \{2,3\})$ is a forbidden pattern.
If $a_{1 3} = 0$,
then $( \{1,3,4\}, \{1,2,3\})$ is a forbidden pattern.
Hence, Case 2.2.2.3 does not happen.
Thus, Case 2.2.2 does not happen. 

\section{Conclusion}
We give a necessary and sufficient condition 
for the connectedness of the solution graph of an integer linear system
in some special cases.
This results is slightly stronger than the
necessary and sufficient condition proposed by 
Shigenobu and Kamiyama~\cite{shigenobuCOCOA2023}.
An apparent future work is to determine whether 
our necessary condition is also sufficient condition in general. 

\bibliographystyle{plain}
\bibliography{ILS_forbidden_pattern}

\end{document}